\documentclass[11pt,english]{article}
\usepackage{lmodern}

\usepackage[T1]{fontenc}
\usepackage[latin9]{inputenc}
\usepackage{geometry}
\usepackage{array}
\usepackage{amsthm}
\usepackage{amsmath}
\usepackage{amssymb}

\makeatletter

\providecommand{\tabularnewline}{\\}

\usepackage{enumitem}		% customizable list environments
\theoremstyle{plain}
\newtheorem{thm}{\protect\theoremname}
  \theoremstyle{definition}
  \newtheorem{defn}[thm]{\protect\definitionname}
  \theoremstyle{plain}
  \newtheorem{prop}[thm]{\protect\propositionname}
  \theoremstyle{plain}
  \newtheorem{fact}[thm]{\protect\factname}
 \theoremstyle{definition}
 \newtheorem*{defn*}{\protect\definitionname}
  \theoremstyle{plain}
  \newtheorem{lem}[thm]{\protect\lemmaname}
  \theoremstyle{plain}
  \newtheorem{cor}[thm]{\protect\corollaryname}
  \theoremstyle{remark}
  \newtheorem*{rem*}{\protect\remarkname}
  \theoremstyle{plain}
  \newtheorem{claim}[thm]{\protect\claimname}
  \theoremstyle{plain}
  \newtheorem*{thm*}{\protect\theoremname}

\makeatother

\usepackage{babel}
  \providecommand{\claimname}{Claim}
  \providecommand{\corollaryname}{Corollary}
  \providecommand{\definitionname}{Definition}
  \providecommand{\factname}{Fact}
  \providecommand{\lemmaname}{Lemma}
  \providecommand{\propositionname}{Proposition}
  \providecommand{\remarkname}{Remark}
  \providecommand{\theoremname}{Theorem}
\providecommand{\theoremname}{Theorem}

\begin{document}
\global\long\def\prob{\mathbb{P}}
\global\long\def\E{\mathbb{E}}
\global\long\def\Bin{\mathrm{Bin}}

\title{On Active and Passive Testing}

%% TODO: fix Rani's affiliation?
\author{Noga Alon%
\thanks{Sackler School of Mathematics and Blavatnik School of
Computer Science, Tel Aviv University, Tel Aviv 69978, Israel and
School of Mathematics, Institute for Advanced Study, Princeton, NJ 08540.
Email: \texttt{nogaa@tau.ac.il}.
Research supported in part by BSF grant 2012/107,
by ISF grant 620/13, by the Israeli I-Core program and by the Fund for
Mathematics.}
\and Rani Hod%
\thanks{School of Mathematics, Georgia Tech, 686 Cherry st, Atlanta, GA 30332.
E-mail: \texttt{rani.hod@math.gatech.edu}. %
} \and Amit Weinstein%
\thanks{Blavatnik School of Computer Science, Tel Aviv University, Tel Aviv
69978, Israel. Email: \texttt{amitw@tau.ac.il}. Research supported
in part by the Israeli 
I-Core program.%
}}
\maketitle
\begin{abstract}
%NA
Given a property of Boolean functions, what is the minimum number
of queries required to determine with high probability if an input
function satisfies this property or is ``far'' from satisfying it?  
This is a fundamental question in Property Testing, where traditionally 
the testing algorithm is allowed to pick its queries among the entire 
set of inputs. Balcan, Blais, Blum and Yang have
recently suggested to restrict the tester to take its queries from
a smaller random subset of polynomial size
of the inputs. This model is called
\emph{active testing}, and in the extreme case when the size of the 
set we can query from is exactly the number of queries 
performed it is known as \emph{passive testing}.

We prove that passive or active testing of 
$k$-linear functions (that is, sums of $k$ variables among $n$ 
over $\mathbb{Z}_{2}$)
requires $\Theta(k\log n)$ queries, assuming $k$ is not too
large. This extends the case $k=1$, (that is, dictator functions),
analyzed by Balcan et. al.

We also consider other classes of functions including
low degree polynomials, juntas, and partially symmetric functions.
Our methods combine algebraic, combinatorial, and 
probabilistic techniques,
including the Talagrand concentration inequality and the Erd\H{o}s--Rado
theorem on $\Delta$-systems.
\end{abstract}

\section{Introduction}

Property testing considers the following general problem: given a
property $\mathcal{P}$, identify the minimum number of queries required
to determine with high probability whether an input object has the
property $\mathcal{P}$ or whether it is \textquotedblleft far\textquotedblright{}
from $\mathcal{P}$. This question was first formalized by Rubinfeld
and Sudan~\cite{RS96} in the context of Boolean functions.
\begin{defn}[\cite{RS96}]
\label{def:membership-tester}Let $\mathcal{P}$ be a family of Boolean
functions and let $\epsilon>0$. A \emph{$q$-query $\epsilon$-tester} for
$\mathcal{P}$ is a randomized algorithm that queries an unknown function
$f:\mathbb{Z}_{2}^{n}\to\mathbb{Z}_{2}$ on $q$ inputs of its choice
and
\begin{enumerate}[label=(\roman*)]
\item Accepts with probability at least 2/3 when $f\in\mathcal{P}$; 
\item Rejects with probability at least 2/3 when $f$ is $\epsilon$-far from
$\mathcal{P}$, where $f$ is $\epsilon$-far from $\mathcal{P}$
if ${\rm dist}(f,g):=\left|\left\{ x\in\mathbb{Z}_{2}^{n}\mid f(x)\neq g(x)\right\} \right|\geq\epsilon2^{n}$
holds for every $g\in\mathcal{P}$.
\end{enumerate}
We denote the minimal $q$ such that a $q$-query $\epsilon$-tester
for $\mathcal{P}$ exists by $Q_{\epsilon}\left(\mathcal{P}\right)$.
\end{defn}
The main line of research in many works on property testing is to
characterize $Q_{\epsilon}\left(\mathcal{P}\right)$ for various properties~$\mathcal{P}$.
An interesting distinction is identifying properties for which $Q_{\epsilon}\left(\mathcal{P}\right)$
is constant (i.e., independent of $n$). For instance, linearity can
be tested in a constant number of queries~\cite{BLR93}; more generally,
testing if a Boolean function is a polynomial of constant degree can
be performed with a constant number of queries~\cite{AKKLR03,BFL91,BKSSZ10,RS96}.
Testing whether a function depends only on a constant number of its
input variables (that is, if a function is a junta) can also be done
with a constant number of queries~\cite{Blais08,Blais09,FKRSS04}.

\medskip{}

In the definition above the algorithm can pick its $q$ queries in
the entire set $\mathbb{Z}_{2}^{n}$. Balcan, Blais, Blum, and
Yang~\cite{BBBY13}
suggested to restrict the tester to take its queries from a smaller,
typically random, subset $U\subseteq\mathbb{Z}_{2}^{n}$. This model
is called active testing, in resemblance of active learning (see,
e.g.,~\cite{CAL94}). Active testing gets more difficult as the size
of $U$ decreases, and the extreme case is when $U$ is a set of $q$
random points (so the algorithm actually has no choice). This is known
as passive testing, or testing from random examples%
\footnote{Although the examples could be drawn from a general probability distribution,
in this work we focus on the uniform distribution.%
}, and was studied in~\cite{GGR98,KR00}. Formally, the next definition
from~\cite{BBBY13} extends Definition~\ref{def:membership-tester}
to active and passive testers.
\begin{defn}
Let $\mathcal{P}$ be a family of Boolean functions and let $\epsilon>0$.
A \emph{$u$-sample $q$-query $\epsilon$-tester} for $\mathcal{P}$ is
a randomized algorithm that, given a subset $U\subseteq\mathbb{Z}_{2}^{n}$
of size $\left|U\right|=u$, drawn uniformly at random, queries an
unknown function $f:\mathbb{Z}_{2}^{n}\to\mathbb{Z}_{2}$ on $q$
inputs from $U$ and
\begin{enumerate}[label=(\roman*)]
\item Accepts with probability at least 2/3 when $f\in\mathcal{P}$; 
\item Rejects with probability at least 2/3 when $f$ is $\epsilon$-far from
$\mathcal{P}$.
\end{enumerate}

%NA
The set $U$ may be chosen with or without repetitions. For our 
purpose these two options will be equivalent, as in the
parameters considered here the probability of a 
repetition is negligible.
\vspace{0.3cm}

We denote by $Q_{\epsilon}^{a}\left(\mathcal{P},u\right)$ the minimal
$q$ such that a $u$-sample $q$-query $\epsilon$-tester for $\mathcal{P}$
exists 
%NA
($\infty$ if $u$ queries do not suffice), 
and by $Q_{\epsilon}^{p}\left(\mathcal{P}\right)$ the minimal
$q$ such that a $q$-sample $q$-query $\epsilon$-tester (i.e.,
a passive $\epsilon$-tester) for $\mathcal{P}$ exists.
\end{defn}
We are usually interested in $\mathrm{poly}\left(n\right)$-sample
testers; for simplicity, we omit the sample size $u$ from our notation
when this is the case.\medskip{}

%NA

The following inequality from~\cite{BBBY13} shows the relation between
the query complexity of the different testing models.
\begin{prop}[{\cite[Theorem A.4.]{BBBY13}}]
\label{prop:classic-active-passive}For every property $\mathcal{P}$
and for every $\epsilon>0$, $Q_{\epsilon}\left(\mathcal{P}\right)\le Q_{\epsilon}^{a}\left(\mathcal{P}\right)\le Q_{\epsilon}^{p}\left(\mathcal{P}\right)$.
\end{prop}
To provide a simple upper bound on the query complexity of passive
testing, we refer to the more difficult problem of proper passive
learning. The most common model of passive learning is PAC-learning,
introduced by Valiant~\cite{Valiant84}.
\begin{defn}
Let $\mathcal{P}$ be a family of Boolean functions. A \emph{$q$-query
$\epsilon$-learning algorithm} for $\mathcal{P}$ is a randomized
algorithm that, given $q$ random queries from an unknown function
$f\in\mathcal{P}$, outputs a Boolean function $g:\mathbb{Z}_{2}^{n}\to\mathbb{Z}_{2}$
such that $g$ is $\epsilon$-close to $f$ with probability at least
$2/3$ (the underlying probability space is the random queries and
the coin tosses of the algorithm). The algorithm is called \emph{proper}
if it always returns some $g\in\mathcal{P}$. We denote the minimal
$q$ such that a proper $q$-query $\epsilon$-learning algorithm
for $\mathcal{P}$ exists by $Q_{\epsilon}^{\ell}\left(\mathcal{P}\right)$.
\end{defn}
The number of queries needed to properly learn a Boolean function
essentially bounds from above the number of queries needed to test
it; given the output of a proper learning algorithm, it remains to
verify that the input function is indeed close to it. More formally,
we have the following proposition.
\begin{prop}[{\cite[Proposition 3.1.1]{GGR98}}]
\label{prop:passive-testing-upper-bound-by-learning}For every property
$\mathcal{P}$ and for every $\epsilon>0$, $Q_{\epsilon}^{p}\left(\mathcal{P}\right)\le Q_{\epsilon/2}^{\ell}\left(\mathcal{P}\right)+O\left(1/\epsilon\right)$.
\qed
\end{prop}
This proposition is often used together with the following known upper
bound.
\begin{fact}
\label{fact:simple-PAC-UB}For every family of Boolean functions $\mathcal{P}$,
$Q_{\epsilon}^{\ell}(\mathcal{P})=O(\tfrac{1}{\epsilon}\log|\mathcal{P}|)$.
\qed
\end{fact}
%%Throughout the remainder of this paper, we focus our attention at 

For the sake of simplicity, we focus on a constant $\epsilon$ (say,
$\epsilon=0.001$) throughout the rest of this paper. This allows
us to drop the subscript $\epsilon$ from our notation when possible
(e.g., we write $Q\left(\mathcal{P}\right)$ instead of $Q_{\epsilon}\left(\mathcal{P}\right)$).

\subsection{Our results}

In~\cite{BBBY13} it was shown that active testing of dictator functions
(i.e., functions that only depend on a single input variable) requires
$\Theta\left(\log n\right)$ queries. Our first result extends this
to the family of $k$-linear functions; that is, the family of sums
of $k$ variables over $\mathbb{Z}_{2}$. Let $\mathrm{Lin}_{k}$
denote this family.
\begin{thm}
\label{thm:k-linear-active}Active or passive testing of Boolean $k$-linear
functions requires $\Theta\left(k\log n\right)$ queries, for all
$k \leq \frac{\log n}{10\log\log n}$.
\end{thm}
Theorem~\ref{thm:k-linear-active} and its proof imply a lower bound
for active testing of superfamilies of $k$-linear functions, such
as $k$-juntas and $\left(n-k\right)$-symmetric functions. A function
is called \emph{$k$-junta} if it depends on at most $k$ of its input variables,
referred to as the influential variables (e.g., a dictator function
is a $1$-junta). We denote the family of $k$-juntas by $\mathrm{Jun}_{k}$.
Partially symmetric functions are a generalization of juntas, where
the remaining variables can influence the output of the function,
but only in a symmetric manner.
\begin{defn}[Partially symmetric functions \cite{BWY12}]
For a subset $T\subseteq[n]:=\{1,\ldots,n\}$, a function $f:\mathbb{Z}_{2}^{n}\to\mathbb{Z}_{2}$
is called $T$-\emph{symmetric} if permuting the labels of the variables
of $T$ does not change the function. Moreover, $f$ is called $t$-\emph{symmetric}
if there exists $T\subseteq[n]$ of size at least $t$ such that $f$
is $T$-symmetric. We denote the family of $t$-symmetric functions
by $\textrm{Sym}_{t}$.
\end{defn}
Partially symmetric functions were introduced as part of the research
of isomorphism testing~\cite{BWY12,CFGM12}, where it was shown that
testing whether a function is $(n-k)$-symmetric for any $k<n/10$
can be done using $O(k\log k)$ queries. The special case of $2$-symmetric
functions has already been considered by Shannon in~\cite{Sha49}.
In addition to the $\Omega(k\log n)$ lower bound for active testing
of partially symmetric functions, we provide an upper bound as well
as lower and upper bounds for passive testing (detailed in Table~\ref{tab:main-result}).
In particular, we show that for a constant $k$, the family of partially
symmetric functions demonstrates a significant gap among the three
different testing scenarios and proper learning.
\begin{thm}
\label{thm:psf-separation}For a constant $k$ we have 
\begin{eqnarray*}
Q\left(\mathrm{Sym}_{n-k}\right) & = & \Theta\left(1\right),\\
Q^{a}\left(\mathrm{Sym}_{n-k}\right) & = & \Theta\left(\log n\right),\\
Q^{p}\left(\mathrm{Sym}_{n-k}\right) & = & \Theta(n^{1/4}\sqrt{\log n}),\\
Q^{\ell}\left(\mathrm{Sym}_{n-k}\right) & = & \Theta(\sqrt{n}).
\end{eqnarray*}

\end{thm}
The last family of functions considered in this work is low degree
polynomials, with special consideration given to linear functions.
The following indicates that passive testing of degree $d$ polynomials,
denoted by $\mathrm{Pol}_{d}$, is essentially as hard as properly
learning them.
\begin{thm}
\label{thm:passive-polynoms-UB-and-LB}The query complexity of passive
testing of degree $d$ polynomials is $\Theta(n^{d})$, for constant
$d$.
\end{thm}
On the other hand, active testing can be done slightly more efficiently,
at least for linear functions.
\begin{thm}
\label{thm:active-linear-UB-and-LB}The query complexity of active
testing of linear functions is $\Theta(n/\log n)$.
\end{thm}
Table~\ref{tab:main-result} summarizes the results presented in
this work for passive and active testing, as well as the best known
query complexity for the classical model of property testing and proper
learning.

\begin{table}[h]
\noindent \centering{}%
\begin{tabular}{|>{\centering}m{3cm}|>{\centering}m{2.5cm}|>{\centering}m{2.5cm}|>{\centering}m{3.5cm}|c|}
\hline 
Family & Classic ($Q$) & Active ($Q^{a}$) & Passive ($Q^{p}$) & Learning ($Q^{\ell}$)\tabularnewline
\hline 
\hline 
Symmetric & $O(1)$ & $O(1)$ & $\Theta(n^{1/4})$ & $\Theta(\sqrt{n})$\tabularnewline
\hline 
Linear & $O(1)$ \cite{BLR93} & $\Theta(n/\log n)$ & $n+\Theta(1)$ & $n+\Theta(1)$\tabularnewline
\hline 
$d$-degree polynomials & $\Theta(2^{d})$ \cite{AKKLR03,BKSSZ10} &  & $\Theta(n^{d})$ & $\Theta(n^{d})$\tabularnewline
\hline 
$k$-linear & $O(k\log k)$,\\
$k-o(k)$ \cite{Blais09,BK12} & $\Theta(k\log n)$ & $\Theta(k\log n)$ & $\Theta(k\log n)$\tabularnewline
\hline 
$k$-juntas & $O(k\log k)$,\\
$\Omega(k)$ \cite{Blais09,BBM12,CG04} & $\Omega(k\log n)$ & $\Omega(2^{k/2}+k\log n)$ & $\Theta(2^{k}+k\log n)$\tabularnewline
\hline 
$(n-k)$-symmetric & $O(k\log k)$,\\
$\Omega(k)$ \cite{BWY12} & $O(2^{k}k\log n)$,\\
$\Omega(k\log n)$ & $O(n^{1/4}\sqrt{2^{k}k\log n})$,\\
$\Omega(n^{1/4}\sqrt{2^{k}+k\log n})$ & $\Theta(\sqrt{n}2^{k})$\tabularnewline
\hline 
\end{tabular}\protect\caption{\label{tab:main-result}
Summary of best bounds, for fixed $\epsilon$
and $k<\log n/(10 \log\log n)$}
\end{table}

The rest of the paper is organized as follows. The lower bound for
active testing of $k$-linear functions, which applies to juntas and
partially symmetric functions as well, is proved in Section~\ref{sec:-linear-functions}
by establishing a general result for random subsets of abelian groups,
proved by combining probabilistic and combinatorial tools including
the Talagrand inequality and the Erd\H{o}s--Rado results on $\Delta$-systems.
Section~\ref{sec:Partially-symmetric-functions} provides the lower
and upper bounds for active and passive testing of symmetric and partially
symmetric functions, as described in Table~\ref{tab:main-result}.
The results concerning low degree polynomials and linear functions
in particular are presented in Section~\ref{sec:Low-degree-polynomials}.
Concluding remarks and open problems are in Section~\ref{sec:Discussion}.
The proofs in Sections~\ref{sec:Partially-symmetric-functions} and~\ref{sec:Low-degree-polynomials}
are also based on probabilistic, combinatorial, and algebraic techniques.

\section{$k$-linear functions\label{sec:-linear-functions}}

Theorem~\ref{thm:k-linear-active} states that the query complexity
of active or passive testing of $k$-linear functions is $\Theta(k\log n)$.
The upper bound can be obtained by applying Propositions~\ref{prop:classic-active-passive}
and~\ref{prop:passive-testing-upper-bound-by-learning}, and Fact~\ref{fact:simple-PAC-UB},
given that there are exactly $\binom{n}{k}$ different $k$-linear
functions.

In order to prove a lower bound for active testing of $k$-linear
functions, we use the following lemma, which is an adaptation of the
tools used in~\cite{BBBY13} to prove active testing lower bounds
(specifically, Theorem 6.6 and Lemma B.1 ibid).
\begin{defn*}
A property $\mathcal{P}$ is called \emph{$\epsilon$-nontrivial} if a random
Boolean function is $\epsilon$-close to $\mathcal{P}$ with probability
at most $0.01$.\end{defn*}
\begin{lem}[\cite{BBBY13}]
 \label{lem:active-testing-dimension} Let $\mathcal{P}$ be an $\epsilon$-nontrivial
property of Boolean functions and let $\pi$ be a distribution supported
on $\mathcal{P}$. Given a set $S=\left\{ x_{1},x_{2},\ldots,x_{q}\right\} $
of $q$ queries and a vector $y\in\mathbb{Z}_{2}^{q}$, define 
\[
\pi_{S}(y)=\prob_{f\sim\pi}\left[f\left(x_{i}\right)=y_{i}\mbox{ for }i=1,2,\ldots,q\right].
\]
Choose at random a set $U$ of $u$ samples, and suppose that with
probability at least $\frac{3}{4}$, every set $S\subseteq U$ of
$q$ queries and every $y\in\mathbb{Z}_{2}^{q}$ satisfy $\pi_{S}\left(y\right)<\frac{6}{5}2^{-q}$.
Then, $Q_{\epsilon}^{a}(\mathcal{P},u)\geq q$.
\qed
\end{lem}
The proof is based on the fact that, under the assumptions of the
lemma, $q$ queries do not suffice to distinguish between a function
from the distribution $\pi$ and a uniform random Boolean function.

According to Lemma~\ref{lem:active-testing-dimension}, our goal
is therefore to show that when we choose a random $k$-linear function,
querying it at $o(k\log n)$ queries chosen from a random space will
appear rather random. To this end we use Lemma~\ref{lem:main}, which,
roughly speaking, assures us that the probability of seeing a given
output vector is very concentrated around the expectation. The proof
of the lemma uses the Talagrand inequality (with an extra twist) and
the Erd\H{o}s--Rado $\Delta$-systems method. Lemma~\ref{lem:main},
its proof, and the tools used appear in Section~\ref{sec:proof-of-main-lemma}.

The following theorem provides a lower bound for active testing of
$k$-linear functions, completing the proof of Theorem~\ref{thm:k-linear-active}
(assuming Lemma~\ref{lem:main}).
\begin{thm}
\label{thm:k-linear-active-LB} $Q^{a}\left(\mathrm{Lin}_{k},u\right)=\Omega\left(k\log n\right)$
for $k \leq 0.1 \log n/\log\log n$, as long as $n\le u\le
%\exp(n^{1/7k})$ .
2^{n^{1/7k}}$.
\end{thm}
\begin{proof}
Define $\pi$ to be the uniform distribution over the $k$-linear
functions. In particular, $\pi$ is the distribution obtained by choosing
distinct $i_{1},i_{2},\ldots,i_{k}\in\left[n\right]$ uniformly at
random and returning the function $f:\mathbb{Z}_{2}^{n}\to\mathbb{Z}_{2}$
defined by $f\left(x\right)=x_{i_{1}}+x_{i_{2}}+\cdots+x_{i_{k}}$.
Fix $S$ to be a set of $q$ vectors in $\mathbb{Z}_{2}^{n}$. This
set can be viewed as a $q\times n$ Boolean-valued matrix. We write
$c_{1}\left(S\right),\ldots,c_{n}\left(S\right)$ to represent the
columns of this matrix. For any $y\in\mathbb{Z}_{2}^{q}$,
\[
\pi_{S}\left(y\right)=\binom{n}{k}^{-1}\left|\left\{ I\in\binom{\left[n\right]}{k}:\sum_{i\in I}c_{i}\left(S\right)=y\right\} \right|\ .
\]

Since $\mathrm{Lin}_{k}$ is, say, $0.4$-nontrivial, by Lemma~\ref{lem:active-testing-dimension},
to prove that $Q^{a}(\mathrm{Lin}_{k},u)=\Omega\left(k\log n\right)$,
it suffices to show that when $U$ is a set of $u$ vectors chosen
uniformly and independently at random from $\mathbb{Z}_{2}^{n}$ and,
say, $q=\left(1-\frac{1}{k}\right)\log\binom{n}{k}+k$, then with
probability at least $\frac{3}{4}$, every set $S\subset U$ of size
$\left|S\right|=q$ and every $y\in\mathbb{Z}_{2}^{q}$ satisfy $\pi_{S}\left(y\right)\le\frac{6}{5}2^{-q}$.
To this end, we would like to show that $\pi_{S}(y)$ is highly concentrated
around $\E[\pi_{S}(y)]=2^{-q}$.

To apply Lemma~\ref{lem:main}, consider the group $G=\mathbb{Z}_{2}^{q}$
and let $N=\left|G\right|=2^{q}=2^{k}\binom{n}{k}^{1-1/k}$. 
By monotonicity,
we assume $u=\lfloor 2^{n^{1/7k}}\rfloor\ge n$ 
and let $\lambda=\lceil qn^{1/7k}\rceil\ge q\log u$.
Now, for large enough $n$ conditions~(\ref{eq:lemma-condition-a}) and~(\ref{eq:lemma-condition-b})
of the lemma hold. Indeed, to prove the first inequality note that
\begin{equation}
\label{e91}
800\ln2\cdot kN\lambda^{2k+1}
 = 800 \ln 2 \binom{n}{k}^{1-1/k} k2^{k}\lceil q n^{1/7k}\rceil ^{2k+1}
 \le 
800 \binom{n}{k}^{1-1/k} \cdot (2qn^{1/7k})^{2k+1}.
\end{equation}
Since $k<0.1 \log n/\log \log n$ and $q < 2k \log n < (\log n)^2$,
we have
$$
800 (2qn^{1/7k})^{2k+1} =o(\sqrt n)<n/k <\binom{n}{k}^{1/k}.
$$
Therefore, the right-hand-side of (\ref{e91}) is smaller than
$\binom{n}{k}$, establishing (\ref{eq:lemma-condition-a}).

To prove the second inequality note that
\begin{equation}
\label{e92}
\frac{\lambda N}{k 2^k} =\frac{\lambda \binom{n}{k}^{1-1/k}}
{k} \geq \frac{q n^{1/7k} \binom{n}{k}}{\binom{n}{k}^{1/k} k}
=\frac{(n-k+1) q n^{1/7k}}{k \binom{n}{k}^{1/k} k}
\binom{n}{k-1}.
\end{equation}
However, 
$$
\frac{(n-k+1) q n^{1/7k}}{k \binom{n}{k}^{1/k} k}
\geq \Omega\left( \frac{n k \log n }{nk} \right) >1,
$$
and therefore the right-hand-side of (\ref{e92}) is bigger than
$\binom{n}{k-1}$, proving (\ref{eq:lemma-condition-b}).

Thus, for any fixed vector $y\in\mathbb{Z}_{2}^{q}$,
the probability that more than $\frac{6}{5}\binom{n}{k}2^{-q}$ $k$-sets
of columns of $S$ sum to $y$ is at most $5\cdot2^{-\lambda}$. Furthermore,
when $U$ is defined as above, we can apply the union bound over all
$y\in G$ and over all subsets $S\subseteq U$ of size $\left|S\right|=q$
to obtain 
\[
\prob\left[\exists S,y:\pi_{S}\left(y\right)>\frac{6}{5}2^{-q}\right]\le\binom{u}{q}\cdot2^{q}\cdot5\cdot2^{-\lambda}\le\frac{u^{q}}{q!}\cdot2^{q}\cdot5\cdot2^{-q\log u}=o\left(1\right),
\]
establishing the theorem.
\end{proof}
The above theorem and its proof immediately imply a lower bound for
active testing of both $k$-juntas and $(n-k)$-symmetric functions.
This can also be applied to show lower bounds for other concise representation
families, such as small DNF formulas, small decision trees, small
Boolean formulas, and small Boolean circuits (see \cite{DLMORSW07}).
\begin{cor}
\label{cor:juntas-active-LB}$Q^{a}(\mathrm{Jun}_{k})=\Omega(k\log n)$
and $Q^{a}(\mathrm{Sym}_{n-k})=\Omega(k\log n)$ for $k=O\left(\log n/\log\log n\right)$.\end{cor}
\begin{proof}
The same distribution $\pi$ from the proof of Theorem~\ref{thm:k-linear-active-LB}
(uniform distribution over the $k$-linear functions) is supported
on $k$-juntas (resp., $\left(n-k\right)$-symmetric functions) and
these properties, too, are still $0.4$-nontrivial.
\end{proof}
In Section~\ref{sec:Partially-symmetric-functions} we continue the
investigation of active and passive testing of partially symmetric
functions. The following proposition summarizes what we know about
passive testing of $k$-juntas.
\begin{prop}
\label{prop:juntas-passive-UB-and-LB} $\Omega(2^{k/2}+k\log n)\le Q^{p}(\mathrm{Jun}_{k})\le O\left(2^{k}+k\log n\right)$.\end{prop}
\begin{proof}
The upper bound is obtained by applying Proposition~\ref{prop:passive-testing-upper-bound-by-learning}
and Fact~\ref{fact:simple-PAC-UB}, as the number of $k$-juntas
is $\binom{n}{k}2^{2^{k}}$. The lower bound is a combination of two
separate bounds:%
\footnote{Although Corollary~\ref{cor:juntas-active-LB} only holds for $k=O(\log n/\log\log n)$,
for larger values of $k$ its contribution to the lower bound is negligible.%
} $\Omega\left(k\log n\right)$ by Corollary~\ref{cor:juntas-active-LB}
and $\Omega\left(2^{k/2}\right)$ for verifying that the input function
is indeed a junta, even when the set of the influencing variables
is known in advance. Indeed, assume we are given the input function
with a promise that it is either a random junta over the first $k$
variables or a random function. Distinguishing between these two cases
is impossible unless we have a pair of inputs agreeing on the first
$k$ variables;  among less than $\frac{1}{2}2^{k/2}$ queries, we
get such a pair with probability at most $\frac{1}{2}\left(\frac{1}{2}2^{k/2}\right)^{2}\cdot2^{-k}=1/8$.
\end{proof}

\subsection{Proof of main lemma\label{sec:proof-of-main-lemma}}

Before we state the formal lemma, we introduce the two following combinatorial
and probabilistic tools used in the proof.

\paragraph{Erd\H{o}s--Rado $\Delta$-systems}
\begin{defn*}
Let $a,b$ be positive integers. We say that a family of $a$ sets,
each of size $b$, forms a \emph{$\Delta$-system} of size $a$ if
all pairs have the same intersection.
\end{defn*}
Erd\H{o}s and Rado proved that every large enough family of sets contains
a large $\Delta$-system.
\begin{thm}[{\cite[Theorem 3]{ER60}}]
\label{thm:sunflower}Let $\mathcal{F}$ be a family of sets, each
of size $b$. Then $\mathcal{F}$ contains $\Delta$-system of size
$a$ whenever $\left|\mathcal{F}\right|\ge\left(a-1\right)^{b+1}b!$.
\qed
\end{thm}

\paragraph{Talagrand's concentration inequality\protect \\
}

In its general form, Talagrand's inequality is an isoperimetric inequality
for product probability spaces. We use the following formulation from~\cite{MR-book}
(see also~\cite{AS-book,Tal95}), suitable for showing that a random
variable in a product space is concentrated around its expectation
under two conditions:
\begin{thm}[{\cite[page 81]{MR-book}}]
\label{thm:talagrand-ineq}Let $X\ge0$ be a non-trivial random variable,
determined by $n$ independent trials $T_{1},\ldots,T_{n}$. If there
exist $c,r>0$ such that
\begin{enumerate}[label=(\roman*)]
\item $X$ is $c$-Lipschitz: changing the outcome of one trial can affect
$X$ by at most $c$, and
\item $X$ is $r$-certifiable: for any $s$, if $X\ge s$ then there is
a set of at most $rs$ trials whose outcomes certify that $X\ge s$,
\end{enumerate}
then for any $0\le t\le\E\left[X\right]$,
\[
\prob\left[\left|X-\E\left[X\right]\right|>t+60\sqrt{\tau}\right]<4\exp\left(-t^{2}/8\tau\right),
\]
where $\tau=c^{2}r\E\left[X\right]$.
\qed
\end{thm}
We now state the main lemma. 
\begin{lem}
\label{lem:main}Let $G$ be an abelian group of order $N$, and let
$n\in\mathbb{N}$. Consider a random sequence $X=\left(x_{1},x_{2},\ldots,x_{n}\right)$,
where each $x_{i}\in G$ is chosen uniformly and independently at
random (with repetitions). Fix $y\in G$ and $k\in\mathbb{N}$, and
let $Y=\left|\mathcal{Y}\right|$, where \textup{$\mathcal{Y}=\left\{ I\in\binom{\left[n\right]}{k}:\sum_{i\in I}x_{i}=y\right\} $}.
Let $\lambda\ge2\log N$ be a positive integer and assume that\begin{subequations}
\begin{alignat}{1}
\binom{n}{k} & \ge800\ln2\cdot kN\lambda^{2k+1};\mbox{ and}\label{eq:lemma-condition-a}\\
\binom{n}{k-1} & \le\frac{\lambda N}{k2^{k}}.\label{eq:lemma-condition-b}
\end{alignat}
\end{subequations}Then, 
\[
\mathrm{\prob}\left[\left|Y-\E\left[Y\right]\right|>\frac{1}{5}\E\left[Y\right]\right]<5\cdot2^{-\lambda}.
\]
\end{lem}
\begin{proof}
For $k=1$ we have $Y\sim{\rm Bin}\left(n,1/N\right)$ and the result
is implied by Chernoff's inequality, so we henceforth assume $k\ge2$.
We would like to use Talagrand's inequality to prove that $Y$ is
concentrated around $\E\left[Y\right]$, but $Y$ does not satisfy
the Lipschitz condition necessary for its application. Let us thus
define $\hat{Y}=\left|\hat{\mathcal{Y}}\right|$, where $\hat{\mathcal{Y}}\subseteq\mathcal{Y}$
is maximal such that, for all $j\in\left[n\right]$, $x_{j}$ belongs
to at most $c$ sets $I\in\hat{\mathcal{Y}}$; the exact value of
$c$ will be determined later.

First we bound the probability that $\hat{Y}\neq Y$. Let $\mathcal{Y}_{j}=\left\{ I\in\binom{\left[n\right]\setminus\left\{ j\right\} }{k-1}:I\cup\left\{ j\right\} \in\mathcal{Y}\right\} $.
By Theorem~\ref{thm:sunflower}, there exists a $\Delta$-system
$\mathcal{Z}_{j}\subseteq\mathcal{Y}_{j}$ of size 
\[
\left|\mathcal{Z}_{j}\right|\ge1+\sqrt[k]{\left|\mathcal{Y}_{j}\right|/\left(k-1\right)!}>\left|\mathcal{Y}_{j}\right|^{1/k}e/k,
\]
where every two distinct $I_{1},I_{2}\in\mathcal{Z}_{j}$ have the
same intersection $K_{j}=I_{1}\cap I_{2}$. Thus, $\mathcal{Z}_{j}'=\left\{ I\setminus K_{j}:I\in\mathcal{Z}_{j}\right\} $
is a collection of $s_{j}=\left|\mathcal{Z}'_{j}\right|=\left|\mathcal{Z}_{j}\right|$
disjoint $k'$-sets such that $\sum_{i\in I}x_{i}=z$ for all $I\in\mathcal{Z}_{j}'$,
where $k'=k-1-\left|K_{j}\right|\le k-1$ and $z=y-x_{j}-\sum_{i\in K_{j}}x_{i}$.

Consider the event $E_{z}\left(s\right)$, defined as the existence
of a collection of $s$ disjoint $k'$-subsets of $X$ that all sum
to the same element $z\in G$. Then,
\[
\prob\left[E_{z}\left(s\right)\right]\le\binom{n}{\underbrace{k',k',\ldots,k'}_{s}}N^{-s}\le\frac{1}{s!}\binom{n}{k'}^{s}N^{-s}\le\left(\frac{e}{sN}\binom{n}{k-1}\right)^{s}\le\left(\frac{e\lambda}{sk2^{k}}\right)^{s},
\]
and thus we have, for the choice of $c=\lambda^{k}$,

\begin{equation}
\begin{aligned}\prob\left[Y>\hat{Y}\right] & \le\prob\left[\exists j\in\left[n\right]:\left|\mathcal{Y}_{j}\right|>c\right]\le\prob\left[\exists j\in\left[n\right]:s_{j}>c^{1/k}e/k\right]\\
 & \le\prob\left[\exists z\in G:E_{z}\left(e\lambda/k\right)\right]\le N\left(2^{-k}\right)^{e\lambda/k}=2^{\log N-e\lambda}<2^{-2\lambda}.
\end{aligned}
\label{eq:almost-surely-equal}
\end{equation}
This also serves to show that $\E\left[Y\right]$ and $\E[\hat{Y}]$
are very close, since 
\begin{equation}
\E\left[Y-\hat{Y}\right]\le\max\left(Y-\hat{Y}\right)\cdot\prob\left[Y>\hat{Y}\right]\le\binom{n}{k}2^{-2\lambda}\le\binom{n}{k-1}^{2}2^{-2\lambda}\le\left(\frac{\lambda N}{k2^{k}}\right)^{2}2^{-2\lambda}\le\frac{\lambda^{2}2^{-\lambda}}{64}<\frac{1}{32}.\label{eq:same-expectancy}
\end{equation}

Next we apply Talagrand's inequality to bound the deviation of $\hat{Y}$
from $\E[\hat{Y}]$. By definition, $\hat{Y}$ is $c$-Lipschitz;
moreover, to prove that $\hat{Y}\ge s$ we only need to reveal $s$
$k$-sets, i.e., reveal $x_{i}$ for at most $ks$ values of $i$.
For every choice of $I\in\binom{\left[n\right]}{k}$, $\sum_{i\in I}x_{i}$
is a random element of $G$ and thus $\E\left[Y\right]=\binom{n}{k}/N\ge800\ln2\cdot k\lambda^{2k+1}$.

Set $\tau=c^{2}k\E[\hat{Y}]$. By Theorem~\ref{thm:talagrand-ineq},
\begin{equation}
\begin{aligned}\prob\left[\left|\hat{Y}-\E[\hat{Y}]\right|>\frac{1}{10}\E[Y]+60\sqrt{\tau}\right] & \le4\exp\left(-\E[Y]^{2}/800\tau\right)\\
 & \le4\exp\left(-\E[Y]/800c^{2}k\right)<4\exp\left(-\lambda^{2k+1}\ln2/c^{2}\right)=4\cdot2^{-\lambda}.
\end{aligned}
\label{eq:from-talagrand}
\end{equation}
Putting~(\ref{eq:almost-surely-equal}),~(\ref{eq:same-expectancy})
and~(\ref{eq:from-talagrand}) together, 
\begin{align*}
\prob\left[\left|Y-\E\left[Y\right]\right|>\frac{1}{5}\E\left[Y\right]\right] & \le\prob\left[Y>\hat{Y}\right]+\prob\left[\left|\hat{Y}-\E\left[Y\right]\right|>\frac{1}{5}\E\left[Y\right]\right]\\
 & \le2^{-\lambda}+\prob\left[\left|\hat{Y}-\E\left[\hat{Y}\right]\right|>\frac{1}{5}\E\left[Y\right]-\frac{1}{32}\right]<5\cdot2^{-\lambda},
\end{align*}
under the condition $\frac{1}{5}\E\left[Y\right]-\frac{1}{32}\ge\frac{1}{10}\E\left[Y\right]+60\sqrt{\tau}$,
satisfied whenever $\lambda\ge650$.
\end{proof}

\section{Partially symmetric functions\label{sec:Partially-symmetric-functions}}

A key concept in the study of symmetric and partially symmetric functions
is the following notion:
\begin{defn*}[{\cite[Definition 3.1]{BWY12}}]  %% Definition 4 in FOCS
The symmetric influence of a set $T\subseteq\left[n\right]$ of variables
in a Boolean function $f:\mathbb{Z}_{2}^{n}\to\mathbb{Z}_{2}$ is
defined as
\[
{\rm SymInf}_{f}\left(T\right)=\prob_{x\in\mathbb{Z}_{2}^{n},\sigma\in S_{n}}\left[f\left(x\right)\neq f\left(\sigma\left(x\right)\right)\mid\forall i\not\in T:\sigma\left(i\right)=i\right].
\]

\end{defn*}
By definition, a $T$-symmetric function $f$ has ${\rm SymInf}_{f}\left(T\right)=0$;
conversely, for functions far from being $T$-symmetric we have the
following lemma:
\begin{lem}[{\cite[Lemma 3.3]{BWY12}\label{lem:syminf}}]  %% Lemma 3 in FOCS
If $f$ is $\epsilon$-far from being $T$-symmetric, then ${\rm SymInf}_{f}\left(T\right)\ge\epsilon$.
\qed
\end{lem}
In other words, distinguishing between a $T$-symmetric function and
one far from being $T$-symmetric can be done by estimating the symmetric
influence.

\medskip{}

The following proposition determines the number of queries needed
for passive and active testing of symmetric Boolean functions. Although
these results are a special case of partially symmetric functions,
we feel that this serves as an introduction and provides some intuition.
\begin{prop}
\label{prop:symmetric-passive-active-UB-and-LB} $Q^{p}(\mathrm{Sym}_{n})=\Theta(n^{1/4})$
and $Q^{a}(\mathrm{Sym}_{n})=O(1)$.\end{prop}
\begin{proof}
A symmetric function is characterized by its layers of different Hamming
weight. For each Hamming weight between 0 and $n$, the function outputs
a consistent value. To test symmetry given a function, it suffices
to randomly choose an input $x\in\mathbb{Z}_{2}^{n}$ and a permutation
of it, and see if the output is consistent over the two inputs. Since
the Hamming weight of $x$ is distributed $\Bin\left(n,1/2\right)$,
two random inputs share the same Hamming weight with probability $4^{-n}\binom{2n}{n}=\left(1+o\left(1\right)\right)\sqrt{2/\pi n}$;
having fewer than $\frac{1}{2}n^{1/4}$ random samples yields even
a single such pair with probability at most $\frac{1}{2}\left(\frac{1}{2}n^{1/4}\right)^{2}\cdot\sqrt{2/\pi n}<1/8$.
%%On the other hand, among $\left(2\pi n\right)^{1/4}$ random samples
%%we expect $\frac{1}{2}\sqrt{2\pi n}\cdot\sqrt{2/\pi n}=1$ such pair.
%%By Markov, repeating this $12/\epsilon$ times results in at least
%%$2/\epsilon$ pairs with probability $5/6$.

On the other hand, among $4\left(2\pi n\right)^{1/4}$ random samples,
it is not hard to see that the probability of not having 
such a pair is smaller than, say, $1/7$.
%NA
(One way to show this fact is by looking for matches 
between the first and 
second halves of the samples,
assuming the first half did not yield such a pair already. In this case
with high probability the total measure of the layers in which 
we have a representative from the first half is at least
$\frac{1}{(2 \pi n)^{1/4}}$ and conditioning on this, the probability
that no sample from the second half falls into one of these layers
is smaller than $e^{-2}.$)

By Markov, repeating this $14/\epsilon$ times results in at least $12/\epsilon$ sets
without a desired pair with probability at most $1/6$. Therefore, with probability at least $5/6$,
we have at least $2/\epsilon$ pairs.
By Lemma \ref{lem:syminf},
if the function is $\epsilon$-far from being symmetric then each
such pair will have different outputs with probability at least $\epsilon$,
so we will fail to detect this with probability $\left(1-\epsilon\right)^{2/\epsilon}<1/e^{2}$.
Altogether the success probability exceeds $2/3$.

In the context of active testing, given a sample space of, say, $u=n/\epsilon$
vectors we can easily find $2/\epsilon$ input pairs with the same
Hamming weight each, thus testing whether the input function is indeed
symmetric can be done using $4/\epsilon$ queries.\end{proof}
\begin{rem*}
Consider the following slight modification of the algorithms above.
Instead of rejecting the input function upon the first example of
it not being symmetric, we estimate its symmetric influence by counting
the number of such examples among all pairs. This enables us to passively
(resp., actively) distinguish between a function that is $\epsilon/2$-close
to being symmetric and one that is $\epsilon$-far using $O\left(\epsilon^{-2}n^{1/4}\right)$
(resp., $O\left(\epsilon^{-2}\right)$) queries. Such an algorithm
is called a \emph{tolerant tester}.
\end{rem*}
Some families of Boolean functions, such as symmetric and partially
symmetric functions, have many pairs of functions which are close
to one another. In these cases, the upper bound of Fact~\ref{fact:simple-PAC-UB},
which relies only on the size of the family, is not tight. We remedy
this by proving the following refined version.
\begin{defn}
Let $\mathcal{P}$ be a family of Boolean functions and let $\epsilon>0$.
Denote by $\mathcal{I}_{\epsilon}\left(\mathcal{P}\right)$ a subfamily
of $\mathcal{P}$ of maximal size such that every two distinct $f,g\in\mathcal{I}_{\epsilon}\left(\mathcal{P}\right)$
are $\epsilon$-far.\end{defn}
\begin{prop}
\label{prop:refined-PAC-UB-LB}Let $\mathcal{P}$ be a family of Boolean
functions and let $\epsilon>0$. Then 
\[
\left\lfloor \log\left|\mathcal{I}_{2\epsilon}\left(\mathcal{P}\right)\right|\right\rfloor \leq Q_{\epsilon}^{\ell}\left(\mathcal{P}\right)\le\left\lceil \tfrac{64}{\epsilon}\ln\left|\mathcal{I}_{\epsilon/2}\left(\mathcal{P}\right)\right|\right\rceil .
\]
\end{prop}
\begin{proof}
A proper learning algorithm for $\mathcal{P}$ is required to return
a function from $\mathcal{P}$ that is $\epsilon$-close to the input
function. Since functions in $\mathcal{I}_{2\epsilon}\left(\mathcal{P}\right)$
are $2\epsilon$-far from one another, the algorithm has to return
a different output for each of them. Any deterministic algorithm making
$q$ queries can only have $2^{q}$ different outputs, so if it performs
less than $\left\lfloor \log\left|\mathcal{I}_{2\epsilon}\left(\mathcal{P}\right)\right|\right\rfloor $
queries, it must be wrong with probability at least $1/2$. A randomized
algorithm for this problem can be viewed as a distribution over deterministic
algorithms (as the queries are chosen randomly and the algorithm is
non-adaptive), and therefore cannot improve the success probability
beyond $1/2$.

Next, consider the following learning algorithm: given an input function
$f\in\mathcal{P}$, return the function $g\in\mathcal{I}_{\epsilon/2}\left(\mathcal{P}\right)$
that agrees with $f$ on as many queries as possible out of $q=\left\lceil \mbox{\ensuremath{\left(64/\epsilon\right)}}\ln\left|\mathcal{I}_{\epsilon/2}\left(\mathcal{P}\right)\right|\right\rceil $
random queries. By definition, $f$ is $\epsilon/2$-close to some
$f'\in\mathcal{I}_{\epsilon/2}\left(\mathcal{P}\right)$; therefore,
$f$ and $f'$ disagree on each query with probability at most $\epsilon/2$,
independently. The total number of disagreements is thus dominated
by a $\Bin\left(q,\epsilon/2\right)$ random variable and hence with
high probability they disagree on fewer than $3\epsilon q/4$ queries.
Using a similar argument, a function $h\in\mathcal{I}_{\epsilon/2}\left(\mathcal{P}\right)$
that is $\epsilon$-far from $f$ will disagree with $f$ on more
than $3\epsilon q/4$ queries with probability at least $1-\exp\left(-\epsilon q/32\right)=1-\left|\mathcal{I}_{\epsilon/2}\left(\mathcal{P}\right)\right|^{-2}$.
By the union bound, with high probability no such $h$ will outperform
$f'$ and thus the algorithm will return a function that is $\epsilon$-close
to $f$ (the obvious candidate being $f'$).\end{proof}
\begin{cor}
\label{cor:psf-proper-learning-UB-and-LB} $Q^{\ell}(\mathrm{Sym}_{n-k})=\Theta(2^{k}\sqrt{n-k})$
for $k<n$; in particular, $Q^{\ell}(\mathrm{Sym}_{n})=\Theta(\sqrt{n})$.\end{cor}
\begin{proof}
First, we show that $\left|\mathcal{I}_{\epsilon/2}(\mathrm{Sym}_{n-k})\right|=2^{O\left(2^{k}\sqrt{n-k}\right)}$.
The binomial distribution $\Bin\left(n-k,1/2\right)$ is concentrated
around its center, and in particular the middle $\ell=1+2\left\lceil \sqrt{\left(n-k\right)\ln\left(4/\epsilon\right)/2}\right\rceil $
layers account for at least $1-\epsilon/2$ of the weight. In other
words, every $\left(n-k\right)$-symmetric function is $\left(\epsilon/2\right)$-close
to an $\ell$-canonical $\left(n-k\right)$-symmetric function, which
is zero outside the middle $\ell$ layers. We can thus bound $\left|\mathcal{I}_{\epsilon/2}(\mathrm{Sym}_{n-k})\right|$
from above by $2^{2^{k}\ell}$, the number of $\ell$-canonical functions.

For the lower bound, consider the middle $\ell'=1+2\left\lfloor \sqrt{n-k}\right\rfloor $
layers. The weight ratio between any pair of these layers is bounded
by $\binom{n-k}{\left\lfloor \left(n-k\right)/2\right\rfloor }/\binom{n-k}{\left\lfloor \left(n-k\right)/2-\sqrt{n-k}\right\rfloor }<e^{2}$.
Let $\mathcal{C}\subset\mathbb{Z}_{2}^{2^{k}\ell'}$ be an error correcting
code of rate $1/2$ and relative distance $1/10$; in other words,
$\mathcal{C}$ has at least $2^{2^{k-1}\ell'}$ codewords, every pair
of which are $\left(1/10\right)$-far. We can interpret each codeword
as an $\ell'$-canonical $\left(n-k\right)$-symmetric function, which
is $\left(1/10e^{2}\right)$-far from the rest. Hence we get $\left|\mathcal{I}_{2\epsilon}(\mathrm{Sym}_{n-k})\right|\ge2^{2^{k-1}\ell'}$
as long as $\epsilon<1/20e^{2}$.

Therefore, for our fixed $\epsilon$, the result follows from Proposition~\ref{prop:refined-PAC-UB-LB}.
\end{proof}
Proposition~\ref{prop:psf-active-passive-UB} provides an upper
bound for the query complexity of passive and active testing of partially
symmetric functions. Its proof relies on the following simple concentration
claim in which we make no attempt to optimize the estimates.
\begin{claim}
\label{clm:collisions-concentration}
There is an absolute constant $b>0$ such that
for every $c, 0<c<1$ the following
holds.
Let $s$ and $t$ be integers
satisfying $s<t$. Let $P$ be an arbitrary probability distribution
on $t$ bins, where the probability of each bin is at least $c/t$. 
Then, when we throw 
$s$ balls randomly and independently into $t$ bins
according to the probability $P$, the probability of getting
less than $c s^{2}/9t$ collisions%
\footnote{A single collision happens every 
time we place a ball in an already
occupied bin.%
} is at most $\exp\left(-bcs^{2}/t\right)$.
\end{claim}
\begin{proof}
If the number of occupied bins is 
less than $s/3$ after $\left\lceil s/2\right\rceil $
balls were thrown, then we already have at least $s/6>cs^{2}/9t$ 
collisions.
Otherwise, each of the next $\left\lfloor s/2\right\rfloor $ balls
has a probability of at least $cs/3t$ to collide with these occupied
bins, independently. The number of collisions created by the last
$\left\lfloor s/2\right\rfloor $ balls thus dominates a binomial
$\Bin\left(\left\lfloor s/2\right\rfloor ,cs/3t\right)$ random variable.
By Chernoff, it is less than $cs^{2}/9t$ with probability
at most $\exp\left(-bc s^{2}/t\right)$.\end{proof}
\begin{prop}
\label{prop:psf-active-passive-UB} 
$Q^{p}(\mathrm{Sym}_{n-k})=O\left(n^{1/4}2^{k/2}\sqrt{k\log n}\right)$
and $Q^{a}(\mathrm{Sym}_{n-k})=O\left(2^{k}k\log n\right)$, for $k=o\left(\log n\right)$.\end{prop}
\begin{proof}
We begin with a passive testing algorithm. Let $f$ be the tested
Boolean function. Our algorithm asks $q=d(\epsilon)
n^{1/4}2^{k/2}\sqrt{k\log n}$
queries, and if the results obtained are consistent with $f$ being
$(n-k)$-symmetric it accepts, otherwise it rejects. It remains to
show that if $f$ is $\epsilon$-far from being $(n-k)$-symmetric
the algorithm rejects with high probability. Assume this is the case
and fix a $k$-set $T\in\binom{\left[n\right]}{k}$ of variables.
If we choose a random vector $x$ and another random vector $y$ obtained
from $x$ by permuting the elements in $[n]\setminus T$ the probability
that $f(x)\neq f(y)$ is at least $\epsilon$ by Lemma~\ref{lem:syminf}.
By Claim~\ref{clm:collisions-concentration} (where each bin
corresponds to the ordered pair consisting of the 
projection on $T$ and the Hamming weight of a typical vector, which
is within distance $\Theta(\sqrt n)$ from $n/2$),
for an appropriately chosen $d(\epsilon)$, with probability
at least $1-n^{-k}$ our queries will contain more 
than $0.5\frac{d(\epsilon)^{2}}{9}k\log n>k\log n/\epsilon$
random disjoint pairs $x,y$ which have the same Hamming weight and agree on
$T$. The probability that none of these pairs will satisfy $f(x)\neq f(y)$
is at most $(1-\epsilon)^{k\log n/\epsilon}<n^{-k}$. The union bound
thus completes the argument.

The same argument implies that the query complexity of active testing
is $O\left(2^{k}k\log n\right)$ because the only queries  the
passive algorithm above actually used are the results for the 
$\Theta(q^{2}/\sqrt{n})=\Theta\left(2^{k}k\log n\right)$
pairs $x,y$ which agree on their Hamming weight. The active
algorithm will thus simply select from the sample $\Theta\left(2^{k}k\log
n\right)$ disjoint pairs with the same Hamming weight  and proceed
as the passive algorithm.
\end{proof}
The following proposition provides a lower bound for the query complexity
of passive testing of partially symmetric functions. Note that it
matches the upper bound, up to a constant factor, when $k$ is constant.
\begin{prop}
\label{prop:psf-passive-LB} $Q^{p}(\mathrm{Sym}_{n-k})=\Omega\left(n^{1/4}\left(2^{k/2}+\sqrt{k\log n}\right)\right)$.\end{prop}
\begin{proof}
As in the proof of Proposition~\ref{prop:juntas-passive-UB-and-LB},
we use a combination of two lower bounds. The first one, $\Omega(n^{1/4}2^{k/2})$,
is required even when the identity of the $k$ asymmetric variables
is known in advance. Assuming we are given the promise that the input
function is either $(n-k)$-symmetric and the asymmetric variables
are the first $k$ variables, or it is far from being partially symmetric,
one still needs to verify the partial symmetry. The only way to verify
it is by having pairs of inputs that share Hamming weight and agree
on the values of the first $k$ variables. However, we expect to see
no such pairs if the number of queries is $o\left(n^{1/4}2^{k/2}\right)$.

The second part of the lower bound uses the $\Omega(k\log n)$ bound
of Theorem~\ref{thm:k-linear-active-LB}. We wish to show that distinguishing
the sum of a random $k$-linear function and a random symmetric function
cannot be distinguished from a random function, given $q=o\left(n^{1/4}\sqrt{k\log n}\right)$
queries. Indeed, assume this many queries were performed and denote
by $H\subseteq\left\{ 0,1,\ldots,n\right\} $ the set of Hamming weights
attained by at least two queries. A balls and bins argument shows
that we expect only $o\left(k\log n\right)$ queries whose Hamming
weight lies in $H$. Due to the random symmetric function, the algorithm
cannot extract any information from queries that have a unique Hamming
weight. Say that we reveal to the algorithm the value of the random
symmetric function on $H$. Now, the algorithm has $o(k\log n)$ queries
and it must distinguish between a $k$-linear function and a random
function. Even if the algorithm were allowed to choose which queries
to pick out of the initial set of $q$ queries, the lower bound for
active testing of $k$-linear functions indicates this cannot be done.
\end{proof}
Theorem~\ref{thm:psf-separation} follows from Propositions~\ref{prop:psf-active-passive-UB}
and~\ref{prop:psf-passive-LB} and Corollaries~\ref{cor:juntas-active-LB}
and~\ref{cor:psf-proper-learning-UB-and-LB}, as well as the results
of~\cite{BWY12}.

\section{Low degree polynomials\label{sec:Low-degree-polynomials}}

We prove Theorem~\ref{thm:passive-polynoms-UB-and-LB} for a more
general case, allowing $1\leq d\leq n^{1/3}$. Let $\binom{n}{\le d}=\sum_{i=0}^{d}{n \choose i}$
be the number of monomials of degree at most $d$. Note that for constant
$d$, we have $\binom{n}{\le d}=\Theta(n^{d})$.
\begin{thm*}[Restatement of Theorem~\ref{thm:passive-polynoms-UB-and-LB}]
$Q^{p}(\mathrm{Pol}_{d})=\Theta(\binom{n}{\le d})$.
\end{thm*}
\begin{proof}
The number of polynomials of degree $d$ is $2^{\binom{n}{\le d}}$,
hence by Fact~\ref{fact:simple-PAC-UB} and Proposition~\ref{prop:passive-testing-upper-bound-by-learning},
passive testing can be done using $O(\binom{n}{\le d})$ queries.
We now show a lower bound of $\Omega(\binom{n}{\le d})$ queries.

Let $x_{1,}x_{2,}\ldots,x_{q}\in\mathbb{Z}_{2}^{n}$ be the set of
$q=\left\lfloor \binom{n}{\le d}/2e\right\rfloor $ random queries
performed by a passive tester. For $i=1,\ldots,q$, define $y_{i}\in\mathbb{Z}_{2}^{\binom{n}{\le d}}$
to be the $d$-evaluation of $x_{i}$, that is, the evaluations of
all possible monomials of degree at most $d$ at $x_{i}$. It suffices
to show that $\text{\ensuremath{\left\{  y_{i}\right\} } }_{i=1}^{q}$
are most likely linearly independent to conclude that any testing
algorithm performs badly; indeed, since the $\binom{n}{\le d}$ monomials
serve as a basis to $\mathrm{Pol}_{d}$, $\text{\ensuremath{\left\{  y_{i}\right\} } }_{i=1}^{q}$
being linearly independent implies that every possible output $\left(f(x_{1}),\ldots,f(x_{q})\right)\in\mathbb{Z}_{2}^{q}$
is equally likely when choosing a random $f\in\mathrm{Pol}_{d}$,
so the tester sees a uniform distribution and therefore cannot decide.

In order to show that, with high probability, these vectors are linearly
independent, we bound the probability that $y_{i}$ is spanned by
$y_{1},\ldots,y_{i-1}$, and then apply the union bound to show that
none of these events is likely to occur. Let $V_{i}=\mathrm{span}\left\{ y_{1},\ldots,y_{i-1}\right\} $
be the linear space spanned by the first $i-1$ vectors. By Lemma~4
from~\cite{BHL12}, since 
\[
\dim V_{i}\le i-1<q\le\binom{n}{\le d}/2e\le\sum_{i=0}^{d}\binom{\left\lceil n(1-1/d)\right\rceil }{i},
\]
no more than $2^{\left\lceil n(1-1/d)\right\rceil }$ $d$-evaluations
of vectors from $\mathbb{Z}_{2}^{n}$ reside in $V_{i}$. Thus, $\prob\left[y_{i}\in V_{i}\right]\le2^{-\left\lceil n/d\right\rceil }$
and, by the union bound, $\prob\left[\exists i:y_{i}\in V_{i}\right]\le q\cdot2^{-\left\lceil n/d\right\rceil }=o(1)$
for $d\le n^{1/3}$.
\end{proof}
We now focus on linear functions, for which we determine the passive
query complexity up to an additive constant term. We slightly abuse
notation by using $\mathrm{Pol}_{1}$ to denote the family of linear
functions, even though degree $1$ polynomials include both linear
and affine functions.
\begin{prop}
\label{prop:linear-passive-UB-and-LB}$Q^{p}(\mathrm{Pol}_{1})=n+\Theta(1)$.\end{prop}
\begin{proof}
As in the proof of Theorem~\ref{thm:passive-polynoms-UB-and-LB},
a linearly independent query set is useless for the testing algorithm.
Let $x_{1,}x_{2,}\ldots,x_{q}$ be a sequence of $q\le n$ queries
and define $X_{i}$ to be the event that $x_{i}\in\mathrm{span}\left\{ x_{1},\ldots,x_{i-1}\right\} $.
The probability that some linear dependency exists among the $q$
queries is

\[
\prob\left[\bigcup_{i=1}^{q}X_{i}\right]
=\prob\left[\bigcup_{i=1}^{q}\left(X_{i}\setminus
\bigcup_{j=1}^{i-1}X_{j}\right)\right]
=\sum_{i=1}^{q}\prob\left[X_{i}\setminus\bigcup_{j=1}^{i-1}X_{j}\right]
%NA
\leq \sum_{i=1}^{q}2^{i-1-n}=\frac{2^{q}-1}{2^{n}}.
\]
For $q>n$, surely any set of $q$ queries is linearly dependent.

Given the computation above, a set of $q\leq n-2$ queries is expected
to be linearly dependent with probability smaller than $1/4$. On
the other hand, $n+O(1)$ queries are very likely to provide a basis
for $\mathbb{Z}_{2}^{n}$ and $O\left(1\right)$ linear dependencies,
so we can learn the unique linear function consistent with the basis
and then verify it; if the function is $\epsilon$-far from linear,
each additional query is inconsistent with the learned function with
a constant probability.
\end{proof}
Active testing allows us to reduce the query complexity by a logarithmic
factor, in comparison to passive testing. We first prove the following
lemma, which is an extension of the analysis of the BLR test provided
by Bellare et al.~\cite{BCHKS96}.
\begin{lem}
\label{lem:BLR-k}Given a function $f:\mathbb{Z}_{2}^{n}\to\mathbb{Z}_{2}$
that is $\epsilon$-far from being linear,
\[
\prob_{x_{1},x_{2},\ldots,x_{2k}\in\mathbb{Z}_{2}^{n}}\left[f(x_{1})+\cdots+f(x_{2k})=f(x_{1}+\cdots+x_{2k})\right]\leq\tfrac{1}{2}+\tfrac{1}{2}(1-2\epsilon)^{2k-1}\ 
\]
\end{lem}
\begin{proof}
Since $f$ is $\epsilon$-far from being linear, when writing it in
the Fourier basis $f\left(y\right)=\sum_{S\subseteq\left[n\right]}\hat{f}\left(S\right)\sum_{i\in S}y_{i}$
all of its Fourier coefficients $\{\hat{f}(S):S\subseteq[n]\}$ are
bounded from above by $1-2\epsilon$. Similar to the analysis for
the case $k=1$, the success probability of this test is:
\[
\tfrac{1}{2}+\tfrac{1}{2}\sum_{S\subseteq[n]}\hat{f}(S)^{2k+1}\le\tfrac{1}{2}+\tfrac{1}{2}\left(\max_{S\subseteq[n]}\hat{f}(S)^{2k-1}\right)\sum_{S\subseteq[n]}\hat{f}(S)^{2}=\tfrac{1}{2}+\tfrac{1}{2}\max_{S\subseteq[n]}\hat{f}(S)^{2k-1}\le\tfrac{1}{2}+\tfrac{1}{2}(1-2\epsilon)^{2k-1},
\]
where the middle equality holds by Parseval's theorem.
\end{proof}
Unlike the BLR test, which uses the case $k=1$, in the context of
active testing we need $k$ to be almost linear in $n$, hence little
amplification is necessary.
\begin{thm*}[Restatement of Theorem~\ref{thm:active-linear-UB-and-LB}]
$Q^{a}(\mathrm{Pol}_{1},u)=\Theta(n/\log u)$, for $u\ge n^{2}$.\end{thm*}
\begin{proof}
As done in the previous proof, we bound the number of queries from
below by showing that one is not expected to find a linear dependency
of size smaller than $n/(2\log u)$ among a set of $u$ samples. The
expected number of linear dependencies of size at most $q$ is at most
\[
\sum_{i=0}^q {u \choose i}2^{-n}\leq
u^{q}2^{-n}=2^{q\log u-n}\leq2^{-n/2}\ ,
\]
assuming $q\leq n/(2\log u)$. By Markov's inequality, the probability
of having such a linear dependency is $o\left(1\right)$ and therefore
$\Omega(\frac{n}{\log u})$ queries are needed.

Given an input function that is $\epsilon$-far from being linear,
we use the test of Lemma~~\ref{lem:BLR-k} to identify this. Fix
$q=4 \left\lceil n/\log u\right\rceil $. Given a sample $U$ of $u$
vectors, it contains $\binom{u}{q/2}>2^n$ subsets of size
$q/2$. By the pigeonhole principle two of these sets have the same
sum, hence there is a linear dependency of length at most $q$.
On the other hand, by the previous computation, with high
probability there is no linear dependency of size less than
$n/(2 \log u)=q/8$ hence the length exceeds $q/8$.
By Lemma~\ref{lem:BLR-k} the
probability that $f$ passes a single such test is at most 
$\tfrac{1}{2}+\tfrac{1}{2}(1-2\epsilon)^{q/8-1}<
\tfrac{1}{2}+\tfrac{1}{2}(1-2\epsilon)^{n/2\log u}$.
Since $\epsilon$ is constant, for large enough $n$ this is smaller
than $0.9$, thus repeating the test a constant number of times reduces
the probability of $f$ passing all of them to less than $1/3$ (obviously
we never reject a linear function).
\end{proof}

\section{\label{sec:Discussion}Discussion}

Throughout this work we have demonstrated new bounds for the number
of queries needed for active and passive testing of several properties.
In particular, we now know the amount of queries needed for testing
$k$-linear functions in these new models. 

A practical aspect of property testing algorithms that we did not
cover is the actual running time, rather than just the number of queries
performed, which was the only concern in this work. Some of the algorithms
we presented, especially those based on proper learning, have an exponential
run-time complexity and it would be interesting to see whether active
or passive testing can be done while maintaining polynomial running
time.

Quite a few of the passive testing algorithms we provided can in fact
be made tolerant; that is, they can be modified to accept functions
close to satisfying the property while rejecting functions far from
satisfying it (with some gap in between). For simplicity we did not
explicitly show that. Such modifications usually do not have an effect
on the asymptotic query complexity.

While Section~\ref{sec:Low-degree-polynomials} provides a tight
analysis of active and passive testing of linear functions, for low
degree polynomials our analysis is only tight for passive testing.
Extrapolating based on the behavior of linear functions, it seems
natural to expect that the query complexity of active testing of low
degree polynomials is asymptotically lower than passive testing, perhaps
by a polylogarithmic factor. This question remains open at the moment.

Finally we mention that Lemma~\ref{lem:main}, used in the proof
of Theorem~\ref{thm:k-linear-active-LB}, can be used in the study
of a seemingly unrelated problem of exhibiting a very sharp cutoff
phenomenon in the mixing time of random walks in random (dense) Cayley
graphs of abelian groups. Indeed, the lemma implies that for any abelian
group $G$ of order $N$, and for $(\log N)^{1/3}\le k\le(\log N)^{1/2-\delta}$,
if we choose $d\approx N^{1/(k-1)}$ random elements of $G$, then
a random walk of length $k-1$ in the resulting Cayley graph of $G$
is far from being mixing (simply because we cannot reach most of the
elements at all) while a random walk of length $k$ is already mixing.
While it is more interesting to study this problem for much sparser
random Cayley graphs (see~\cite{LS10} for some related results),
even the above statement for the dense case is interesting.


\begin{thebibliography}{25}
\bibitem{AKKLR03}Noga Alon, Tali Kaufman, Michael Krivelevich, Simon
Litsyn, and Dana Ron, Testing low-degree polynomials over GF(2). In:
\emph{Proceedings of the 6th International Workshop on Approximation
Algorithms for Combinatorial Optimization Problems and 7th International
Workshop on Randomization and Approximation Techniques in Computer
Science (RANDOM--APPROX '03)}, pp.~188--199, 2003.

\bibitem{AS-book}Noga Alon and Joel H. Spencer, \textbf{The Probabilistic
Method} (3rd Edition). Wiley, New York, 2008.

\bibitem{BBBY13}Maria-Florina Balcan, Eric Blais, Avrim Blum, and
Liu Yang, Active property testing. In: \emph{Proceedings of the 53rd
Annual IEEE Symposium on Foundations of Computer Science (FOCS '12)},
pp.~21--30, 2012.

\bibitem{BCHKS96}Mihir Bellare, Don Coppersmith, Johan H{\aa}stad, Marcos
Kiwi, and Madhu Sudan, Linearity testing in characteristic two. \emph{IEEE
Transactions on Information Theory}~\textbf{42}(6): 1781--1796, 1996.

\bibitem{BFL91}L\'{a}szl\'{o} Babai, Lance Fortnow, and Carsten Lund. Nondeterministic
exponential time has two-prover interactive protocols. \emph{Computational
Complexity} \textbf{1}(1):3--40, 1991

\bibitem{BHL12}Ido Ben-Eliezer, Rani Hod, and Shachar Lovett, Random
low-degree polynomials are hard to approximate. \emph{Computational
Complexity}~\textbf{21}(1): 63--81, 2012.

\bibitem{BKSSZ10}Arnab Bhattacharyya, Swastik Kopparty, Grant Schoenebeck,
Madhu Sudan, and David Zuckerman, Optimal testing of Reed-Muller codes.
In: \emph{Proceedings of the 51st Annual IEEE Symposium on Foundations
of Computer Science (FOCS '10)}, pp.~488--497, 2010.

\bibitem{Blais08}Eric Blais, Improved bounds for testing juntas.
In: \emph{Proceedings of the 11th International Workshop on Approximation
Algorithms for Combinatorial Optimization Problems and 12th International
Workshop on Randomization and Approximation Techniques in Computer
Science (RANDOM--APPROX '08)}, pp.~317--330, 2008.

\bibitem{Blais09}Eric Blais, Testing juntas nearly optimally. In:
\emph{Proceedings of the 41st Annual ACM Symposium on Theory of Computing
(STOC' 09)}, pp.~151--158, 2009.

\bibitem{BBM12}Eric Blais, Joshua Brody, and Kevin Matulef, Property
Testing Lower Bounds via Communication Complexity. \emph{Computational
Complexity} \textbf{21}(2):311--358, 2012.

\bibitem{BK12}Eric Blais and Daniel Kane, Tight bounds for testing
$k$-linearity. In: \emph{Proceedings of the 15th International Workshop
on Approximation Algorithms for Combinatorial Optimization Problems
and 16th International Workshop on Randomization and Approximation
Techniques in Computer Science (RANDOM--APPROX '12)}, pp.~435--446,
2012.

\bibitem{BWY12}Eric Blais, Amit Weinstein, and Yuichi Yoshida, Partially
symmetric functions are efficiently isomorphism-testable.
\emph{Proceedings
of the 53rd Annual IEEE Symposium on Foundations of Computer Science
(FOCS '12)}, pp.~551--560, 2012.
Also:
\emph{SIAM J.~on Computing}~\textbf{44}(2):411--432, 2015.

\bibitem{BLR93}Manuel Blum, Michael Luby, and Ronitt Rubinfeld, Self-testing/correcting
with applications to numerical problems. In: \emph{J.~of Computer
and System Sciences~}\textbf{47}:549--595, 1993.

\bibitem{CFGM12}Sourav Chakraborty, Eldar Fischer, David Garc\'{i}a--Soriano,
and Arie Matsliah, Junto-symmetric functions, hypergraph isomorphism,
and crunching. In: \emph{Proceedings of the 27th Annual IEEE Conference
on Computational Complexity (CCC '12)}, pp.~148--158, 2012.

\bibitem{CG04}Hana Chockler and Dan Gutfreund, A lower bound for
testing juntas. \emph{Information Processing Letters}~\textbf{90}(6):301--305,
2004.

\bibitem{CAL94}David Cohn, Les Atlas, and Richard Ladner, Improving
generalization with active learning. In: \emph{Proceedings of the
15th International Conference on Machine Learning (ICML '94)}, pp.~201--221,
1994.

\bibitem{DLMORSW07}Ilias Diakonikolas, Homin Lee, Kevin Matulef,
Krzysztof Onak, Ronitt Rubinfeld, Rocco Servedio, and Andrew Wan.
Testing for concise representations. In: \emph{Proceedings of the
48th Annual IEEE Symposium on Foundations of Computer Science (FOCS
\textquoteright 07)}, pp. 549\textendash 558, 2007.

\bibitem{ER60}Paul Erd\H{o}s and Richard Rado, Intersection theorems
for systems of sets. \emph{J.~London Math.~Soc.}~\textbf{35}(1):
85\textendash 90, 1960.

\bibitem{FKRSS04}Eldar Fischer, Guy Kindler, Dana Ron, Shmuel Safra,
and Alex Samorodnitsky, Testing juntas. \emph{J.~of Computer and
System Sciences}~\textbf{68}(4):753--787, 2004.

\bibitem{GGR98}Oded Goldreich, Shafi Goldwasser, and Dana Ron, Property
testing and its connection to learning and approximation. J.of the
ACM~\textbf{45}(4):653--750, 1998.

\bibitem{KR00}Michael Kearns and Dana Ron, Testing problems with
sublearning sample complexity. \emph{J.~of Computer and System Sciences}~\textbf{61}(3):428--456,
2000.

\bibitem{LS10}Eyal Lubetzky and Allan Sly, Cutoff phenomena for random
walks on random regular graphs. \emph{Duke Mathematical Journal~}\textbf{153}(3):475--510,
2010.

\bibitem{MR-book}Michael Molloy and Bruce Reed, \textbf{Graph colouring
and the probabilistic method}, Springer, 2002.

\bibitem{RS96}Ronitt Rubinfeld and Madhu Sudan, Robust characterizations
of polynomials with applications to program testing. \emph{SIAM J.~on
Computing}~\textbf{25}(2):252--271, 1996.

\bibitem{Sha49}Claude E.~Shannon, The synthesis of two-terminal
switching circuits. \emph{Bell System Technical Journal}~\textbf{28}(1):59--98,
1949.

\bibitem{Tal95}Michel Talagrand, Concentration of measure and isoperimetric
inequalities in product spaces. \emph{Publications Math\'{e}matiques de
l'IH\'{E}S}~\textbf{81}:73--203, 1995.

\bibitem{Valiant84}Leslie G.~Valiant, A theory of the learnable.
\emph{Communications of the ACM}~\textbf{27}(11):1134--1142, 1984.\end{thebibliography}
\end{document}